\documentclass{article}

\usepackage{arxiv}
\usepackage{amsmath}
\usepackage{amsthm}
\usepackage{amsfonts}
\usepackage{amssymb}
\usepackage{stackengine}
\usepackage{tikz}
\usetikzlibrary{automata, positioning, arrows, trees}
\usepackage{stmaryrd}
\usepackage{subcaption}
\usepackage{wrapfig}
\usepackage{dsfont}
\usepackage[font=small]{caption}
\usepackage[linesnumbered,ruled,vlined]{algorithm2e}
\usepackage{multirow}
\usepackage{multicol}
\usepackage{cancel}
\usepackage{url} 
\pdfoutput=1

\tikzset{%
    ->,  
    >=stealth, 
    node distance=2.2cm, 
    every state/.style={thick, fill=gray!10}, 
    initial text=$ $, 
    }

\newcommand{\cla}[1]{\left[#1\right]}
\newcommand{\ccla}[1]{\ensuremath{\llbracket #1 \rrbracket}}

\newcommand{\eqdef}{\ensuremath{\triangleq}}
\newcommand{\iffdef}{\ensuremath{\stackrel{\vartriangle}{\iff}}}

\newtheorem{prop}{Proposition}[section]

\newtheorem{cor}{Corollary}[section]

\newcommand{\ra}{\ensuremath{\rightarrow}}
\newcommand{\indist}{\ensuremath{\equiv}}

\newcommand{\proba}[1]{\ensuremath{\probP_{#1}}}

\newcommand{\Real}{\ensuremath{\mathbb{R}}}

\newcommand{\indicator}{\ensuremath{\mathds{1}}}

\newcommand{\Realp}{\ensuremath{\Real_{+}}}

\newcommand{\Stasim}{\ensuremath{E}}

\newcommand{\probP}{\ensuremath{P}}
\newcommand{\Sta}{\ensuremath{Q}}

\newcommand{\NTree}{\ensuremath{T}}

\newcommand{\Tok}{\ensuremath{O}}

\newcommand{\Aut}{\ensuremath{\mathcal{A}}}
\newcommand{\AutB}{\ensuremath{\mathcal{B}}}
\newcommand{\AutG}{\ensuremath{\mathcal{G}}}

\newcommand{\lmodel}{\ensuremath{\mathcal{L}}}
\newcommand{\lmodelT}{\ensuremath{\mathcal{L}_{\terminal}}}

\newcommand{\sta}{\ensuremath{q}}
\newcommand{\worw}{\ensuremath{w}}
\newcommand{\woru}{\ensuremath{u}}
\newcommand{\worv}{\ensuremath{v}}
\newcommand{\worx}{\ensuremath{x}}
\newcommand{\rthp}{\ensuremath{p}}

\newcommand{\rthr}{\ensuremath{r}}

\newcommand{\Symb}{\ensuremath{\Sigma}}
\newcommand{\Psimplex}{\ensuremath{\Delta}}

\newcommand{\accstr}{\ensuremath{\alpha}}
\newcommand{\symb}{\ensuremath{\sigma}}
\newcommand{\policy}{\ensuremath{\pi}}
\newcommand{\tra}{\ensuremath{\tau}}
\newcommand{\emptyW}{\ensuremath{\lambda}}
\newcommand{\quantp}{\ensuremath{\kappa}}
\newcommand{\pdist}{\ensuremath{\rho}}

\newcommand{\pathT}{\ensuremath{\zeta}}

\newcommand{\ce}{\ensuremath{\gamma}}

\newcommand{\QuaNT}{\ensuremath{\mathrm{QNT}}}
\newcommand{\terminal}{\ensuremath{\$}}
\newcommand{\SymbT}{\ensuremath{\Symb_\terminal}}
\newcommand{\staI}{\ensuremath{\sta_{\mathrm{in}}}}

\newcommand{\ProbT}{\ensuremath{\Psimplex(\SymbT)}}
\newcommand{\ProbTok}{\ensuremath{\Psimplex(\Tok)}}
\newcommand{\Words}{\ensuremath{\Symb^\ast}}
\newcommand{\WordsT}{\ensuremath{\SymbT^\ast}}
\newcommand{\policyW}{\ensuremath{\policy^\ast}}

\newcommand{\traW}{\ensuremath{\tra^\ast}}

\newcommand{\EQ}{\ensuremath{\mathbf{EQ}}}
\newcommand{\MQ}{\ensuremath{\mathbf{MQ}}}

\newcommand{\DisS}{\ensuremath{\mathit{Dis}}}
\newcommand{\AccS}{\ensuremath{\mathit{Acc}}}
\DeclareMathOperator{\Pref}{Pref}

\newcommand{\lca}{\ensuremath{\mathsf{lca}}}
\newcommand{\sift}{\ensuremath{\mathsf{sift}}}

\newcommand{\topr}[1]{\ensuremath{\mathsf{top}_{#1}}}
\newcommand{\rank}[1]{\ensuremath{\mathsf{rank}_{#1}}}

\newcommand{\reach}{\ensuremath{\mathit{reach}}}

\newcommand{\build}{\ensuremath{\mathsf{build}}}
\newcommand{\update}{\ensuremath{\mathsf{update}}}

\newcommand{\qstaI}{\ensuremath{\overline{\sta}_{\mathrm{in}}}}

\newcommand{\qpolicy}{\ensuremath{\overline{\policy}}}
\newcommand{\qtra}{\ensuremath{\overline{\tra}}}
\newcommand{\qsta}{\ensuremath{\overline{\sta}}}
\newcommand{\qSta}{\ensuremath{\overline{\Sta}}}
\newcommand{\qAut}{\ensuremath{\overline{\Aut}}}
\newcommand{\qAutB}{\ensuremath{\overline{\AutB}}}

\newcommand{\bos}{\ensuremath{\textsc{\tiny BOS}}}
\newcommand{\eos}{\ensuremath{\textsc{\tiny EOS}}}

\newcommand{\supp}{\ensuremath{\textsf{supp}}}
\newcommand{\tokenizer}{\ensuremath{\textsf{tok}}}
\newcommand{\samp}{\ensuremath{\textsf{samp}}}

\newcommand{\isdef}[1]{\ensuremath{\indicator_{#1}}}
\newcommand{\indistICGI}{\ensuremath{\indist^\bullet}}
\newcommand{\cclaICGI}[1]{\ensuremath{\ccla{#1}}^\bullet}
\newcommand{\samptopr}[1]{\ensuremath{\mathsf{samptop}_{#1}}}
\newcommand{\qlmodel}{\ensuremath{\overline{\lmodel}}}
\newcommand{\Str}{\ensuremath{\mathrm{Char}^\ast}}
\newcommand{\Tokseq}{\ensuremath{\Tok^\ast}}
\newcommand{\tostr}{\ensuremath{\mathsf{str}}}
\newcommand{\tokenize}[1]{\ensuremath{\stackon[-.1pt]{#1}{$\scriptscriptstyle\frown$}}}
\newcommand{\probaP}{\ensuremath{\boldsymbol{P}}}

\newcommand{\charsymb}[1]{\textsl{\small #1}}
\newcommand{\sync}{\ensuremath{\times}}

\newcommand{\isundef}{\ensuremath{\boldsymbol{0}}}
\newcommand{\qhole}{\ensuremath{{\isundef}}}
\newcommand{\pref}{\ensuremath{\mathsf{pref}}}

\title{Analyzing constrained LLM through PDFA-learning\thanks{Presented at LearnAut 2024, July 7, 2024.}}
\renewcommand{\shorttitle}{Analyzing constrained LLM through PDFA-learning}

\author{
    {M. Carrasco \hspace*{1.5ex}
    F. Mayr \hspace*{1.5ex}
    S. Yovine \hspace*{1.5ex}
    J. Kidd \hspace*{1.5ex}
    M. Iturbide \hspace*{1.5ex}
    J. da Silva \hspace*{1.5ex}
    A. Garat}\\
    Facultad de Ingeniería\\ Universidad ORT Uruguay\\ Montevideo, Uruguay\\
    \texttt{carrasco\_m@ort.edu.uy}\\
    \texttt{mayr@ort.edu.uy}\\
    \texttt{yovine@ort.edu.uy}\\
}

\renewcommand{\headeright}{M. Carrasco et al.}

\begin{document}

\maketitle

\setlength{\abovedisplayskip}{0pt}
\setlength{\belowdisplayskip}{0pt}
\setlength{\abovedisplayshortskip}{0pt}
\setlength{\belowdisplayshortskip}{0pt}

\begin{abstract}
We define a congruence that copes with null next-symbol probabilities that arise when the output of a language model is constrained by some means during text generation. We develop an algorithm for efficiently learning the quotient with respect to this congruence and evaluate it on case studies for analyzing statistical properties of LLM. 
\end{abstract}


\section{Introduction}\label{sec:intro}
Many works have studied neural language models, such as Recurrent Neural Networks (RNN) and Transformers, through the analysis of surrogate automata of different sorts obtained from the former in a variety of ways, with the purpose of verifying or explaining their behavior (e.g.~\cite{empirical_eval_rule_extraction,Weiss17,bollig_21,DBLP:conf/icgi/MayrYCPV23,pmlr-v217-muskardin23a}).
A few have proposed to somehow compose neural language models with automata or regular expressions in order to verifying properties on-the-fly while learning~(\cite{mayr2021property}), assessing the existence of memorization, bias, or toxicity~(\cite{kuchnik2023validating}), and guiding text generation~(\cite{willard2023efficient}).

In this paper, we first study theoretical questions that arise when applying this last approach in the context of active learning of probabilistic deterministic finite automata (PDFA)~(\cite{IEEE:Vidal2005}). In Sec.~\ref{sec:lmodels}, we address the question of dealing with null next-symbol probabilities that appear when constraining the output of a language model by composing it with an automaton and/or a sampling strategy, such as the top $k$ most likely symbols. We do this by defining an appropriate congruence that induces a quotient PDFA without 0-probability transitions. In Sec.~\ref{sec:omit-zero}, we adapt the learning algorithm of~\cite{DBLP:conf/icgi/MayrYCPV23} to efficiently learn the quotient PDFA. In Sec.~\ref{sec:llm}, we discuss issues that arise when analyzing real large language models, in particular the role of tokenizers, and apply the algorithm on problems discussed in~\cite{kuchnik2023validating,willard2023efficient} when generating text with GPT2. Experimental results show the interest of our approach.

\section{Language models}\label{sec:lmodels}
Let \Symb\ be a finite set of \emph{symbols}, \Words\ the set of finite \emph{strings}, $\emptyW\in\Words$ the \emph{empty} string, and
$\SymbT\eqdef\Symb\cup\{\terminal\}$, where $\terminal\not\in\Symb$ is a special symbol used to denote \emph{termination}.
The \emph{probability simplex} over \SymbT\ is $\ProbT \eqdef \{ \pdist:\Symb\ra\Realp \mid \sum_{\symb\in\Symb} \pdist(\symb)=1 \}$.
The \emph{support} of $\pdist\in\ProbT$ is $\supp(\pdist) \eqdef \{ \symb\in\Symb \mid \pdist(\symb)>0\}$.
A \emph{language model} is a total function $\lmodel:\Words\ra\ProbT$.

Language models can be expressed in different ways, e.g., RNN, Transformers, and PDFA. 
Following~\cite{DBLP:conf/icgi/MayrYCPV23}, we define a PDFA \Aut\ over \Symb\ as a tuple $(\Sta, \staI, \policy, \tra)$, where
\Sta\ is a finite set of states,
$\staI \in \Sta$ is the initial state, 
$\policy : \Sta \ra \ProbT$, and 
$\tra : \Sta \times \Symb \ra \Sta$. 
Both \policy\ and \tra\ are total functions.
We define $\traW$ and $\policyW$ as follows:
$\traW(\sta, \emptyW) \eqdef \sta$ and $\traW(\sta, \symb \woru) \eqdef \traW( \tra(\sta, \symb), \woru )$, and
$\policyW(\sta, \woru) \eqdef \policy( \traW( \sta, \woru) )$.
We omit the state if it is \staI\ and write $\traW(\woru)$ and $\policyW(\woru)$.
\Aut\ defines the language model such that $\Aut(\woru) \eqdef \policyW(\woru)$.
Fig.~\ref{fig:pdfa_ex1} gives examples of PDFA. The number below \sta\ is the probability of termination $\policy(\sta)(\terminal)$, and the one associated with an outgoing transition labeled \symb\ corresponds to $\policy(\sta)(\symb)$.

\begin{figure}[htbp]
\centering
    \begin{subfigure}[c]{0.4\textwidth}
            \begin{tikzpicture}\small
                \node[state, initial above] (q0) {\stackanchor{$q_0$}{0}};
                \node[state, right of=q0] (q1) {\stackanchor{$q_1$}{0.4}};
                \node[state, left of=q0] (q2) {\stackanchor{$q_2$}{0.2}};
                \draw   
                    (q0) edge[above] node{$a/0.3$} (q1)
                    (q0) edge[above] node{$b/0.7$} (q2)
                    (q1) edge[loop above] node{$a/0.6$} (q1)
                    (q1) edge[loop below] node{$b/0$} (q1)
                    (q2) edge[loop above] node{$a/0.4$} (q2)
                    (q2) edge[loop below] node{$b/0.4$} (q2);
            \end{tikzpicture}
        \end{subfigure}
    \begin{subfigure}[c]{0.4\textwidth}
            \begin{tikzpicture}\small
                \node[state, initial above] (q0) {\stackanchor{$q_0$}{0}};
                \node[state, right of=q0] (q1) {\stackanchor{$q_1$}{0.4}};
                \node[state, left of=q0] (q2) {\stackanchor{$q_2$}{0}};
                \draw   
                    (q0) edge[above] node{$a/0.3$} (q1)
                    (q0) edge[above] node{$b/0.7$} (q2)
                    (q1) edge[loop above] node{$a/0.6$} (q1)
                    (q1) edge[loop below] node{$b/0$} (q1)
                    (q2) edge[loop above] node{$a/0.5$} (q2)
                    (q2) edge[loop below] node{$b/0.5$} (q2);
            \end{tikzpicture}
        \end{subfigure}
        \caption{PDFA \Aut\ (left) and \AutB\ (right) over $\Symb = \{a, b\}$ with $\staI = q_0$.}
        \label{fig:pdfa_ex1}
\end{figure}

\paragraph{Sampling}
\lmodel\ can be used to generate random strings $\worx\in\Words$ with $\worx_i\sim\lmodel(\worx_{<i})$, for $i\geq 1$, where $\worx_i$ is the $i$-th symbol and $\worx_{<i}=\worx_1\ldots\worx_{i-1}$ 
with $\worx_{<1}\eqdef\emptyW$. 
That is, by sampling the next symbol to concatenate from the distribution of the prefix until the termination symbol is selected. In general, this procedure may not terminate.
Indeed, $\lmodel$ uniquely defines a probability distribution over $\Words\cup\Sigma^\omega$, where $\Sigma^\omega$ denotes the set of all infinite strings.
More formally, let $\proba{}: \Words\ra\Realp$ be: 
$\proba{}(\emptyW) \eqdef 1$ and 
$\proba{}(\woru\symb) \eqdef \proba{}(\woru) \cdot \lmodel(\woru)(\symb)$.
We expect $\proba{}(\worw)$ to represent the probability of $\worw$ being a prefix. We also define $\proba{\terminal}: \Words\ra\Realp$ by $\proba{\terminal}(\woru) \eqdef \proba{}(\woru) \cdot \lmodel(\woru)(\terminal)$. In this case, we expect $\proba{\terminal}(\worw)$ to represent the probability of occurrence of the finite string $\worw$. Proposition \ref{prop:existence_probability} guarantees the existence of a unique probability distribution \probaP\ over $\Words\cup\Symb^\omega$ whose prefix probabilities are given by $P$ and whose restriction to $\Words$ is given by $\proba{\terminal}$: if $\worx$ is a random string in $\Words\cup\Sigma^\omega$ with distribution \probaP, then $\probaP\left\{\worw\in \pref(\worx)\right\} = \proba{}(\worw)$ and $\probaP\left\{\worx=\worw\right\}=\proba{\terminal}(\worw)$. Here $\pref(\worx)$ denotes the set of all prefixes in $\Words$ of the string $\worx$, including $\worx$ itself.
In general, $\proba{\terminal}$ is not a probability distribution over $\Words$ as it may not sum 1 (\cite{IEEE:Vidal2005}).
In fact, $\sum_{\woru\in\Words} \proba{\terminal}(\woru)=1$ iff $\probaP\left\{|\worx|<\infty\right\}=1$.
Necessary and sufficient conditions for termination of the sampling process involve statements about the probabilities of the terminal symbol~(\cite{du2023}). 

Not every occurrence of a zero probability is problematic.
For \Aut\ in Fig.~\ref{fig:pdfa_ex1}, the fact that $\policy_\Aut(\sta_1)(b)=0$ is harmless:  
$\sum_{\woru\in\Words} \proba{\terminal}(\woru)
= 0.3 \cdot 0.4 \sum_{n=0}^\infty 0.6^n
+ 0.7 \cdot 0.2 \sum_{n=0}^\infty 0.8^n = 0.3 + 0.7 
= 1$.
However, for \AutB, $\policy_\Aut(\sta_2)(\terminal)=0$ is troublesome:  
$\sum_{\woru\in a\Words} \proba{\terminal}(\woru) 
+ \sum_{\woru\in b\Words} \proba{\terminal}(\woru)
= 0.3 \cdot 0.4 \sum_{n=0}^\infty 0.6^n = 0.3 \neq 1$.
\AutB\ can actually be obtained from \Aut\ by constraining the set of symbols to sample from to the top-$2$ most likely ones: 
$\topr{2}(\policy_\Aut(\sta_2)) = \{a, b\}$, and normalizing the probabilities.
It results in that no finite string starting with symbol $b$ can be sampled in \AutB\ with distribution \proba{\terminal}.
We deal with the general situation in this paper since the possibility of non-termination in the sampling process is harmless for our purposes.
Using $\topr{\rthr}$ or $\topr{\rthp}$ (most likely symbols with a cumulative probability cutoff of $\rthp$) is usual practice when sampling from an LLM. So, it is relevant to formalize the effect of these constraints on \lmodel.
%
A \emph{sampling strategy} $\samp:\ProbT\ra\ProbT$ is s.t. for all $\pdist\in\ProbT$, $\supp(\samp(\pdist)) \subseteq \supp(\pdist)$.
%
$\samp(\lmodel)$ is the language model obtained by applying \samp\ to $\lmodel(\woru)$ $\forall\woru\in\Words$ and normalizing the probabilities.
In Fig.~\ref{fig:pdfa_ex1}, $\AutB = \samptopr{2}(\Aut)$, where $\samptopr{\rthr}(\pdist)(\symb) = \pdist(\symb)$ if $\symb\in\topr{\rthr}(\pdist)$, otherwise is 0.

\noindent
\textbf{Congruences}
\proba{} is used in~\cite{Carrasco_Oncina_1999,IEEE:Vidal2005} to define the  equivalence relation $\indist$ in $\Words$ which is a \emph{congruence}
with respect concatenating a symbol:

\begin{align}\label{def:indist_DH}
\woru \indist \worv 
    &\iffdef \forall \worw\in\Words.\ 
        \frac{\proba{}(\woru\worw)}{\proba{}(\woru)} 
        = 
        \frac{\proba{}(\worv\worw)}{\proba{}(\worv)}
\end{align}
We define $\isdef{\lmodel}:\Words\ra\{0,1\}$ such that $\isdef{\lmodel}(\woru)=1$ iff $\proba{}(\woru)>0$.
%
%
\begin{prop}\label{prop:indist_new}
For all $\woru,\worv\in\Words.\ \woru \indist \worv$ if and only if
\begin{equation}\label{eq:indist_new}
\isdef{\lmodel}(\woru)=\isdef{\lmodel}(\worv)
\text{ and }
\forall \worw\in\Words.\ 
\isdef{\lmodel}(\woru\worw)=\isdef{\lmodel}(\worv\worw)=1 
\implies \lmodel(\woru\worw) = \lmodel(\worv\worw).
\end{equation}
\proof\rm See Appendix~\ref{proof:prop_indist_new}.
\end{prop}
%
Let $\Stasim$ be an equivalence relation in $\Psimplex(\SymbT)$. $\pdist =_{\Stasim} \pdist'$ denotes equivalence, $\cla{\Psimplex(\SymbT)}_{\Stasim}$ and $\cla{\pdist}_{\Stasim}$ the quotient of \Psimplex(\SymbT) and the class of \pdist\ induced by \Stasim\ respectively.
We require:
\begin{align}\label{def:supp}
\supp(\pdist) = \supp(\pdist') &\;\mathrm{whenever }\ \pdist =_{\Stasim} \pdist'
\end{align}
Motivated by (\ref{eq:indist_new}) we generalize (\ref{def:indist_DH}) as follows: $\woru\indist_{\Stasim}\worv$ if and only if
\begin{equation}\label{def:indist_new}
\isdef{\lmodel}(\woru)=\isdef{\lmodel}(\worv)
\text{ and }
\forall \worw\in\Words.\ 
\isdef{\lmodel}(\woru\worw)=\isdef{\lmodel}(\worv\worw)=1 
\implies \lmodel(\woru\worw) =_{\Stasim} \lmodel(\worv\worw).
\end{equation}
We denote $\ccla{\Words}_{\Stasim}$ the set of equivalence classes of $\indist_{\Stasim}$ and $\ccla{\woru}_{\Stasim}$ the class of $\woru$. 
%
Since $\isdef{\lmodel}(\woru)=\isdef{\lmodel}(\worv)$ for all $\woru\indist_\Stasim\worv$, we extend $\isdef{\lmodel}$ to $\ccla{\Words}_{\Stasim}$ and write $\isdef{\lmodel}(\ccla{\woru})$.
%
%
%

%

%
\begin{prop}\label{prop:congruence_stasim}
$\indist_{\Stasim}$ is a congruence: $\forall \woru,\worv\in\Words.\ \woru\indist_{\Stasim}\worv \implies \forall\symb\in\Symb.\ \woru\symb \indist_{\Stasim} \worv\symb$.
\proof\rm See Appendix~\ref{proof:prop_congruence_stasim}.
\end{prop}
%
Let $\indistICGI_\Stasim$ be the congruence in $\Words$ defined in~\cite{DBLP:conf/icgi/MayrYCPV23}:
\begin{align}\label{def:indist_old}
\woru \indistICGI_\Stasim \worv 
    &\iffdef \forall \worw\in\Words.\ 
        \lmodel(\woru\worw) =_\Stasim \lmodel(\worv\worw)
\end{align}
We denote by $\boldsymbol{0}$ the $\indist_\Stasim$-class all $\woru\in\Words$ with $\isdef{\lmodel}(\woru)=0$. 

%
\begin{prop}\label{prop:one_to_one}
There exists a one-to-one map $\phi:\ccla{\Words}_{\Stasim}\setminus\{\boldsymbol{0}\}\ra\cclaICGI{\Words}_{\Stasim}$.
\proof\rm See Appendix~\ref{proof:prop_one_to_one}.
\end{prop}
%
%
\begin{cor}
If $\cclaICGI{\Words}_\Stasim$ is finite then $\ccla{\Words}_\Stasim$ is finite, and $\#\ccla{\Words}_\Stasim \leq \#\cclaICGI{\Words}_\Stasim + 1$.
\end{cor}
%
%
For PDFA, $\indist_\Stasim$ (similarly for $\indistICGI_\Stasim$) can be rephrased over \Sta\ as follows:
$\forall \woru, \worv \in \Words$ 
\begin{align}\label{def:indist_sta}
\traW(\woru) \indist_\Stasim \traW(\worv) &\iffdef \woru\indist_\Stasim\worv
\end{align}
Fig.~\ref{fig:pdfa_ex3}(left) illustrates the difference between $\indist_{\Stasim}$ and $\indistICGI_{\Stasim}$. \Stasim\ is equality. States $\sta_0$, $\sta_1$, and $\sta_2$ are not $\indistICGI_{\Stasim}$-equivalent: $\policy(\sta_2) \neq \policy(\sta_0) = \policy(\sta_1)$, and $\policyW(\sta_0, b) \neq \policyW(\sta_1, b)$. However, $\sta_0 \indist_{\Stasim} \sta_1$ because $\isdef{}(\woru) = 1$ and $\policyW(\sta_0, \woru) = \policyW(\sta_1, \woru)$, for $\woru\in\{a\}^\ast$, and $\isdef{}(\woru) = 0$, for $\woru\in b\Words$.
\begin{prop}\label{prop:counterexample}
Let $\lmodel : \Words\ra\ProbT$, $\woru,\worv\in\Words$ such that $\isdef{\lmodel}(\woru)=\isdef{\lmodel}(\worv)=1$.
For every $\worw\in\Words$ such that $\isdef{\lmodel}(\woru\worw)=1$,
if $\lmodel(\woru\worw)\neq_\Stasim\lmodel(\worv\worw)$, 
then there exists $\worw'\in\pref(\worw)$ such that $\lmodel(\woru\worw')\neq_\Stasim\lmodel(\worv\worw')$, and
$\isdef{\lmodel}(\worv\worw')=1$.
\proof\rm See Appendix~\ref{proof:counterexample}.
\end{prop}
%
%
%
For the sake of readability, we assume hereinafter that, unless stated otherwise, the congruence relation is associated with an equivalence \Stasim\ and omit the subscript.\\

\noindent
\textbf{Quotients}
\indist\ induces a \emph{quotient} $\qlmodel:\ccla{\Words}\ra\cla{\ProbT}$ defined as follows:
$\qlmodel(\ccla{\woru}) \eqdef \cla{\lmodel(\woru)}$.
For a PDFA \Aut, its quotient \qAut\ is $( \qSta, \qstaI, \qpolicy, \qtra )$, where 
$\qSta \eqdef \ccla{\reach(\Sta)}$, with $\reach(\Sta) \eqdef \bigcup_{\woru\in\Words} \traW(\woru)$,
$\qstaI \eqdef \ccla{\staI}$,
$\qpolicy(\ccla{\sta}) \eqdef \cla{\sta}$, and
$\qtra(\ccla{\sta}, \symb) \eqdef \ccla{\tra(\ccla{\sta},\symb)}$ for all $\symb\in\Symb$.

\begin{figure}[htbp]
\centering
    \begin{subfigure}[c]{0.45\textwidth}
            \begin{tikzpicture}\small
                \node[state, initial above] (q0) {\stackanchor{$q_0$}{0.9}};
                \node[state, right of=q0] (q1) {\stackanchor{$q_1$}{0.9}};
                \node[state, left of=q0] (q2) {\stackanchor{$q_2$}{0.1}};
                \draw   
                    (q0) edge[above] node{$a/0.1$} (q1)
                    (q0) edge[above] node{$b/0$} (q2)
                    (q1) edge[bend right=75, above] node{$a/0.1$} (q0)
                    (q1) edge[loop above] node{$b/0$} (q1)
                    (q2) edge[bend left=75, above] node{$a/0.2$} (q0)
                    (q2) edge[loop above] node{$b/0.7$} (q2);
            \end{tikzpicture}
    \end{subfigure}
    \begin{subfigure}[c]{0.45\textwidth}
            \begin{tikzpicture}\small
                \node[state, initial above] (q0) {\stackanchor{$q'_0$}{0.9}};
                \node[state, left of=q0] (q2) {\stackanchor{$q'_2$}{0.1}};
                \draw   
                    (q0) edge[loop right] node{$a/0.1$} (q0)
                    (q0) edge[above] node{$b/0$} (q2)
                    (q2) edge[loop left] node{$a/0.2$} (q2)
                    (q2) edge[loop above] node{$b/0.7$} (q2);
            \end{tikzpicture}
    \end{subfigure}
    \caption{Difference between $\indistICGI_{\Stasim}$ and $\indist_{\Stasim}$}
    \label{fig:pdfa_ex3}
\end{figure}

From (\ref{def:indist_sta}), it follows that each $\qsta\in\qSta$ can be represented by an \emph{access} string \woru\ with $\qsta = \ccla{\traW(\woru)}$. 
Let $\accstr(\qsta)$ be the designated access string of \qsta. W.l.o.g., $\accstr(\qstaI) \eqdef \emptyW$.
Given $\qAut$, we can construct a PDFA $\qAut_\accstr \eqdef ( \qSta, \qstaI, \qpolicy_\accstr, \qtra )$, where for all $\qsta\in\qSta$, 
$\qpolicy_\accstr(\qsta) \eqdef \policyW(\accstr(\qsta))$.
Clearly, all choices of \accstr\ yield isomorphic PDFA that are \indist-equivalent. Thus, unless necessary, we omit \accstr\ and use \qAut\ to refer to any such PDFA.
$\qAut$ is the smallest PDFA which is \indist-equivalent to \Aut.
As an example, let \Aut\ and \AutB\ be the PDFA in Fig.~\ref{fig:pdfa_ex3}(left) and (right), respectively. Since all states of \qAut\ are $\not\indistICGI$, we have that $\qAut_{\indistICGI}$ = \Aut. However, $\qAut_{\indist}$ = \AutB\ because $\sta_0 \indist \sta_1 \not\indist \sta_2$.

Here, it is worth to mention that while the choice of \accstr\ is irrelevant with respect to the congruence, different ones may result in different $\proba{\terminal}$. 
Nevertheless, if \Stasim\ induces convex classes, as is the case for quantization, \rank{}, and \topr{} defined in~\cite{DBLP:conf/icgi/MayrYCPV23}, it is always possible to define $\qpolicy(\qsta)$ as a convex combination of distributions in $\cla{\qpolicy_\accstr(\qsta)}_\Stasim$.
%

\section{Learning algorithm}\label{sec:omit-zero}
\vspace*{-1ex}


Based on the results of Sec.~\ref{sec:lmodels}, we develop the algorithm  \textbf{Omit-Zero}, to learn \indist-minimal PDFA. It is a variant of \QuaNT~(\cite{DBLP:conf/icgi/MayrYCPV23}) that differs in specific steps indicated with boxes in Alg~\ref{alg:quant}.
For lack of space, we focus on these. 
\textbf{Omit-zero} maintains a tree $\NTree$ whose nodes are strings which are partitioned in two sets, namely, $\AccS\subset\Words$ and $\DisS\subset\Words$ of \emph{access} and \emph{distinguishing} strings, respectively. 
\AccS\ is the set of \emph{leafs}. Each $\woru\in\AccS$ is labelled with the distribution $\lmodel(\woru)$. 
\DisS\ is the set of non-leaf nodes. 
Both \AccS\ and \DisS\ contain \emptyW, which is also the root and a leaf of \NTree.
Arcs in \NTree\ are labeled with classes in $\cla{\ProbT}$. 
Every outgoing arc from a non-leaf node is labeled with a different class.  
$\forall\woru\neq\woru' \in \AccS$, the lowest common ancestor, $\worw=\lca(\woru, \woru')$, is such that $\lmodel(\woru \worw) \neq\lmodel(\woru' \worw)$.
A difference with \QuaNT\ is that \NTree\ satisfies the following properties: 

\begin{minipage}[b]{0.4\textwidth}
 \begin{align}
 \forall\woru\in\AccS:\qquad\qquad    & \isdef{\lmodel}(\woru)=1 \label{req:all_leafs_defined} 
\end{align}   
\end{minipage}
\begin{minipage}[b]{0.55\textwidth}
 \begin{align}
    & \forall\worw\in\pathT_{\woru}.\ \worw\neq\emptyW \implies \worw_1\in\supp(\lmodel(\woru)) \label{req:first_symbol_defined}
\end{align}   
\end{minipage}
where $\pathT_{\woru} \subseteq \DisS$ is the path of distinguishing strings from the root to the leaf $\woru$.
Notice that (\ref{req:all_leafs_defined}) implies there is no leaf for the class \isundef\ of undefined strings.
\textbf{Omit-Zero} \build\ is different to \QuaNT\ in the way transitions are added. For all $\woru\in\AccS$ and $\symb\in\Symb$:
\begin{align}\label{omit_zero_build}
\tra(\sta_{\woru}, \symb) &\eqdef 
\begin{cases}
\sta_{\woru'} & \symb\in\supp(\lmodel(\woru)),\ \woru' = \sift(\woru\symb) \\ 
\qhole & \textrm{otherwise}
\end{cases}
\end{align}
If \EQ\ returns a counter example $\ce$, i.e, $\lmodel(\ce)\neq\Aut(\ce)$, it is required to be defined in \Aut:
\begin{align}\label{req:ce_defined}
\forall\ce=\EQ(\Aut,\Stasim)\neq\bot.\ \isdef{\Aut}(\ce)=1
\end{align}

\begin{wrapfigure}[16]{R}{0.47\textwidth}
\vspace*{-1em}
    \begin{minipage}{0.47\textwidth}
        \begin{algorithm}[H]\small        
    
        
        $\Aut \leftarrow \mathrm{InitialHypothesis}(\Stasim)$;
        \fbox{$\ce \leftarrow \EQ(\Aut, \Stasim)$}\;
    
        \If{\ce=$\bot$}{
            \Return \Aut\;
        }
        \fbox{$\NTree \leftarrow  \mathsf{InitializeTree}(\ce, \Stasim)$}\;
    
        \While{$\ce\neq\bot$}{
            \fbox{$\Aut \leftarrow \build(\NTree);$} \;
            \fbox{$\ce \leftarrow \EQ(\Aut, \Stasim)$}\;
        
            \If{$\ce\neq\bot$}{
                \fbox{$\NTree \leftarrow \update(\NTree, \ce, \Stasim)$}\;
            }
        }
        \Return \Aut\;
        \caption{Learning algorithm.}
        \label{alg:quant}
        \end{algorithm}
    \end{minipage}
\end{wrapfigure}

Let $\lmodel(\ce)\neq\Aut(\ce)$. By Req.~\ref{req:ce_defined} and Prop.~\ref{prop:counterexample}, there is some $\ce_{<j}\in\pref(\ce)$ such that 
$\isdef{\lmodel}(\ce_{<j})=1$, $\lmodel(\ce_{<j})\neq\Aut(\ce_{<j})$ and for all $i<j$, $\lmodel(\ce_{<i})=\Aut(\ce_{<i})$.
\textsf{InitializeTree} creates the first instance of \NTree, adding \emptyW\ to \DisS\ as root and as leaf to \AccS, which satisfies $\isdef{\lmodel}(\emptyW)$. \QuaNT\ adds \ce\ to \AccS\ which may not be defined. Instead, \textbf{Omit-Zero} adds $\ce_{<j}$ to \AccS. 
Function \update\ only adds $\ce_{<j}$ to \AccS\ at each call. 
The other operation that could add a leaf to \AccS\ is \textsf{sift-update}, called by \sift\ inside \build\ with $\woru\symb$, where $\woru\in\AccS$ and $\symb\in\supp(\lmodel(\woru))$ by (\ref{req:first_symbol_defined}), thus satisfying $\isdef{\lmodel}(\woru\sigma)$.
Then,
%
\text{\rm\bf Omit-Zero} ensures (\ref{req:all_leafs_defined}). 
Moreover, 
the only operation that adds a string $\worw\neq\emptyW$ to \DisS\ is \update, with $\worw=\ce_j\worw'=\lca(\woru,\ce_{<j})$, for some $\woru$ that was already in \AccS, $\ce_j\in\supp(\woru)$, and $\lmodel(\woru)=\lmodel(\ce_{<j})$ (see definition of \sift\ in~\cite{DBLP:conf/icgi/MayrYCPV23}). By (\ref{req:all_leafs_defined}), $\isdef{\lmodel}(\woru)=\isdef{\lmodel}(\ce_{<j})=1$, so $\ce_j\in\supp(\ce_{<j})$.
Then, \text{\rm\bf Omit-Zero} ensures (\ref{req:first_symbol_defined}). 

\vspace*{-1ex}
\begin{prop}\label{prop:termination}
For any PDFA \Aut, \textbf{Omit-Zero} terminates and computes $\qAut$.
\end{prop}
\vspace*{-1em}
\begin{proof}(Sketch)
Correctness of \QuaNT\ and (\ref{req:all_leafs_defined})-(\ref{req:first_symbol_defined}) imply \textbf{Omit-Zero} computes $\qAut$.
Termination of \QuaNT\ and Prop.\ref{prop:one_to_one} imply \textbf{Omit-Zero} terminates. 
\end{proof}

\paragraph{Performance experiments}\label{ssec:experiments}
We compare \textbf{Omit-Zero} against two instances of \QuaNT, varying the behavior of the teacher: \textbf{Standard} uses Hopcroft-Karp algorithm~\cite{hopcroft_karp} as \EQ\ and \MQ\ as in~\cite{DBLP:conf/icgi/MayrYCPV23}, while \textbf{Teacher-Filter} checks if the string being queried by \MQ\ traverses a 0-probability transition, in which case it identifies it as undefined. \textbf{Omit-Zero} and \textbf{Teacher-Filter} use as \EQ\ an adaptation of Hopcroft-Karp that avoids traversing 0-probability transitions.
The comparison is done by randomly generating PDFA. First, we construct DFA using the algorithm in~\cite{nicaud}, which for a given \emph{nominal} size of $n$ it generates DFA of \emph{actual} reachable size normally distributed around $n$. 
Then, DFA are transformed into PDFA by assigning a random probability distribution over \SymbT\ to every state. A parameter $\theta$ is used to control the probability of a symbol to be 0. 
%


\textbf{Running times as function of $\theta$. }
10 random PDFA with $n=500$ and $|\Symb| = m = 20$ were generated for each $\theta\in[0.9, 1)$, with step $0.02$. Each one was run 10 times for every PDFA using quantization equivalence (\cite{DBLP:conf/icgi/MayrYCPV23}), adapted to satisfy (\ref{def:supp}), with parameter $\quantp = 100$.
Fig.~\ref{fig:exps}(a) shows \textbf{Omit-Zero} has the best performance, with an almost constant but important improvement with respect to \textbf{Teacher-Filter}.
%

\textbf{Running times as function of $n$.~}
We compared the performance on 10 random PDFA with $n = 250, 500, 750, 1000$, and $m=10$, using $\quantp=10$ and $\theta=0.9$. Each algorithm was run 10 times for each PDFA.
Fig.~\ref{fig:exps}(b) shows the median of the execution time curves for $n$.
\textbf{Omit-Zero} is always faster than the other two, achieving a speedup of approximately 24x and 3x with respect to \textbf{Standard} and \textbf{Teacher-Filter}, respectively, for $n=1000$.

\begin{figure}[htbp]
\centering
\begin{subfigure}[c]{0.45\textwidth}
            \includegraphics[width=\textwidth]{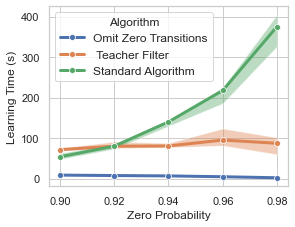}
\end{subfigure}
\hspace*{1ex}
\begin{subfigure}[c]{0.45\textwidth}
            \includegraphics[width=\textwidth]{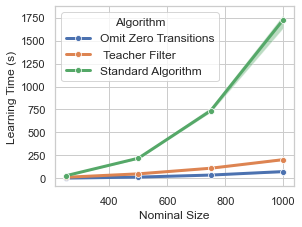}
\end{subfigure}
\caption{Running time curves: %
(left) As function of $\theta$ %
(right) As function of $n$}
\label{fig:exps}
\end{figure}

\vspace*{-3em}
\section{Analyzing large language models}\label{sec:llm}

\paragraph{Guiding generation}\label{ssed:guiding}
Guiding an LLM to generate strings of interest consists in synchronizing it with a automaton that defines the set of symbols that can be drawn at each step of the generation process, which could be constrained further by a sampling strategy. 
To illustrate how the synchronization works, consider the language model given by the PDFA \lmodel\ in Fig.~\ref{fig:synchro} (0-probabilities are omitted). The guide \AutG\ is a \emph{weighted} automaton that defines a \emph{mask} at each state: a weight of 1 for a symbol means it is allowed, otherwise it is not. $\lmodel\sync\AutG$ is a weighted automaton whose underlying structure is the product automaton, and weights are obtained by taking the product of the distribution of the state of \lmodel\ with the weights of the state of \AutG. To obtain PDFA \AutB, we apply the sampling strategy \samptopr{2}. 

\begin{figure}[htbp]
\centering
\begin{subfigure}[c]{0.25\textwidth}
            \begin{tikzpicture}\small
                \node[state, initial left] (q0) {\stackanchor{$q_0$}{$\frac{1}{10}$}};
                \node[state, right of=q0] (q1) {\stackanchor{$q_1$}{$\frac{5}{10}$}};
                \draw   
                    (q0) edge[bend left=75, above] node{$b/\frac{2}{10}$} (q1)
                    (q0) edge[loop above] node{$a/\frac{7}{10}$} (q0)
                    (q1) edge[above] node{$b/\frac{2}{10}$} (q0)
                    (q1) edge[loop above] node{$a/\frac{3}{10}$} (q1);
            \end{tikzpicture}
\end{subfigure}
\begin{subfigure}[c]{0.25\textwidth}
            \begin{tikzpicture}\small
                \node[state, initial above] (q0) {\stackanchor{$q'_0$}{1}};
                \node[state, right of=q0] (q1) {\stackanchor{$q'_1$}{1}};
                \draw   
                    (q0) edge[bend left=75, above] node{$a/1$} (q1)
                    (q1) edge[loop above] node{$a/1$} (q1)
                    (q1) edge[above] node{$b/1$} (q0);
            \end{tikzpicture}
\end{subfigure}
\begin{subfigure}[c]{0.49\textwidth}
            \begin{tikzpicture}\small
                \node[state, initial above] (q0) {\stackanchor{$q_0,q'_0$}{$\frac{1}{8}$}};
                \node[state, right of=q0] (q1) {\stackanchor{$q_0,q'_1$}{0}};
                \node[state, right of=q1] (q2) {\stackanchor{$q_1,q'_0$}{$\frac{5}{8}$}};
                \node[state, right of=q2] (q3) {\stackanchor{$q_1,q'_1$}{$\frac{5}{8}$}};
                \draw   
                    (q0) edge[above] node{$a/\frac{7}{8}$} (q1)
                    (q1) edge[loop above] node{$a/\frac{7}{9}$} (q1)
                    (q1) edge[above] node{$b/\frac{2}{9}$} (q2)
                    (q2) edge[above] node{$a/\frac{3}{8}$} (q3)
                    (q3) edge[loop above] node{$a/\frac{3}{8}$} (q3);
            \end{tikzpicture}
\end{subfigure}
\caption{Synchronization: (left) \lmodel\ (center) \AutG\ (right) $\AutB = \samptopr{2}(\lmodel\sync\AutG)$ }
\label{fig:synchro}
\end{figure}

\paragraph{Learning}\label{ssed:llm_learning}
The teacher knows \lmodel\ and \AutG, while the learner only knows the alphabet of \AutG, and its task is to learn the quotient \qAutB\ of the composition \AutB\ modulo \indist. Notice that in Fig.~\ref{fig:synchro}, \AutB\ is actually not \indist-minimal because $(\sta_1,\sta'_0) \indist (\sta_1,\sta'_1)$. As in~\cite{mayr2021property}, the composition is done \emph{on-demand} during learning. Hence, only \qAutB\ is built.
Moreover, whenever \lmodel\ is an LLM, it is not possible to use as \EQ\ the adapted version of Hopcroft-Karp as done in the experiments in Sec.~\ref{sec:omit-zero}. In this case, Prop.~\ref{proof:counterexample} enables sampling strings doing random walk from the hypothesis constructed by \textbf{Omit-Zero} in order to ensure (\ref{req:ce_defined}).



\paragraph{Tokenizers}\label{ssec:tokenizers}
An LLM, such as GPT2, is a language model whose symbols are usually called \emph{tokens}, denoted \Tok, with $\bos,\eos\in\Tok$ special tokens for \emph{begin} and \emph{end} of sequence. 
To actually query an LLM $\lmodel:\Tokseq\ra\ProbTok$, a string of characters is transformed into a string of tokens by a \emph{tokenizer} $\tokenizer:\Str\ra\Tokseq$.
As an example, consider Huggingface \texttt{Tokenizer}\footnote{\url{https://huggingface.co/docs/transformers/main_classes/tokenizer}}. 
It provides a parameterized tokenizer for various language models. 
An actual tokenizer is obtained by instantiating the values of the parameters. Table~\ref{tab:encoded_gpt2} illustrates the effect of changing the value of parameter \emph{add\_prefix\_space} for GPT2. 
Therefore, in order to guide an LLM with an automaton \AutG, we need to fix \tokenizer\ and also map the symbols \Symb\ of \AutG\ to \Tokseq, by a function $\tostr:\Symb\ra\Str$. We define $\tokenize{\symb} \eqdef \tokenizer(\tostr(\symb))$, and $\tokenize{\terminal} \eqdef \eos$. 
\begin{table}[htbp]\small
  \centering
  \begin{tabular}{|c|c|c|c|c|c|}\hline
    \multirow{2}{*}{\textbf{Symbol}}
    & \multirow{2}{*}{\textbf{\Str}}
    & \multicolumn{2}{|c|}{\textbf{Prefix space}} 
    & \multicolumn{2}{|c|}{\textbf{No prefix space}} \\ \cline{3-6} 
    &
    & \textbf{Tokens} & \textbf{Decoded} 
    & \textbf{Tokens} & \textbf{Decoded} \\ \hline
    \charsymb{medicine} & `medicine' 
    & 9007 & ` medicine' 
    & 1150, 291, 500 & `med', `ic', `ine' \\ \hline
    \multicolumn{6}{}\\
  \end{tabular}
  \caption{Results obtained with two tokenizer instances for GPT2}
  \label{tab:encoded_gpt2}
\end{table}
Now, we must define the probabilities of symbols which are mapped to a sequence of tokens, such as \charsymb{medicine} when \emph{add\_prefix\_space} is false. In this case, we define its probability as the product of the outputs of the LLM for the list of tokens generated by \tokenizer. 
Formally, let $\tokenize{\emptyW} \eqdef \tokenizer(\bos)$, and $\tokenize{\woru\symb} \eqdef \tokenize{\woru}\tokenize{\symb}$. 
$\lmodel_{\tostr,\tokenizer}:\Words\ra\ProbT$ is defined as follows:
\begin{align}\label{def:proba_symb}
\lmodel_{\tostr,\tokenizer}(\woru)(\symb) 
&= \prod\nolimits_{i=1}^{|\tokenize{\symb}|} 
        \lmodel(\tokenize{\woru}\tokenize{\symb}_{<i})(\tokenize{\symb}_i)    
\end{align}

\paragraph{Case study 1}
We run \textbf{Omit-Zero} on GPT2 using the guiding automaton $\AutG_1$ of Fig.~\ref{fig:guiding_pdfa}(a) with \samptopr{2} for both tokenizers. This automaton corresponds to the regex in~\cite{kuchnik2023validating}. The goal is to analyze bias on different professions, namely, medicine, art, computer science, science, information systems, math, engineering, social sciences, humanities, business, after `The man was trained in' and `The woman was trained in'. For convenience \tostr(\charsymb{trained}) is `was trained in'. Table~\ref{tab:man_woman} shows the results obtained for the states of interest in the learnt PDFA, which vary considerably depending on the tokenizer.
%
%
\begin{table}[htbp]
\centering
{\small
\begin{tabular}{|c|lc|lc}\hline
\multirow{2}{*}{\textbf{Access string}}                    
& \multicolumn{2}{c|}{\textbf{With prefix space}}                                  
& \multicolumn{2}{c|}{\textbf{No prefix space}} \\ \cline{2-5} 
& \multicolumn{1}{c|}{\textbf{Symbol 1}} & \multicolumn{1}{c|}{\textbf{Symbol 2}} 
& \multicolumn{1}{c|}{\textbf{Symbol 1}} & \multicolumn{1}{c|}{\textbf{Symbol 2}} \\ \hline
\charsymb{The.man.trained} 
& \multicolumn{1}{c|}{\charsymb{medicine} 0.57}     
& \charsymb{engineering} 0.43                                     
& \multicolumn{1}{c|}{\charsymb{art} 0.72}            
& \multicolumn{1}{c|}{\charsymb{math} 0.28} \\ \hline
\charsymb{The.woman.trained}
& \multicolumn{1}{c|}{\charsymb{medicine} 0.65}        
& \charsymb{business} 0.35                                      
& \multicolumn{1}{c|}{\charsymb{art} 0.80}           
& \multicolumn{1}{c|}{\charsymb{engineering} 0.20} \\ \hline
\multicolumn{5}{}{}
\end{tabular}}
\caption{Probabilities of $\samptopr{2}({GPT2}\times\AutG_1)$ for different tokenizers.}
\label{tab:man_woman}
\end{table}
\vspace*{-1em}
\paragraph{Case study 2}
To study the fidelity of sampling with a learnt PDFA we ran two experiments. First we compare the distributions obtained by guided sampling $10K$ floating points in $[0,1]$ directly on GPT2 and on a PDFA obtained with \textbf{Omit-Zero} by composing GPT2 with the $\AutG_2$ (Fig.~\ref{fig:guiding_pdfa}(b)) that allows only digits $0,\ldots,9$. Second, we use a guiding automaton which allows all 994 numeric tokens of GPT2 and compare the resulting PDFA also with Outlines~\cite{willard2023efficient}. PDFA were obtained using quantization equivalence with $\kappa=100$ and time bounds of 30 and 300 secs, respectively.
%
%
Fig.~\ref{fig:hist_kappa_100} shows the resulting distributions for the first experiment. The $\chi^2$ and Kolmogorov-Smirnov (KS) tests for equality of distributions give the following pvalues: $0.64$ for $\chi^2$ with 10 bins, $0.49$ for $\chi^2$ with 20 bins, and $0.86$ for KS. The KS pvalue for the length distributions is $0.99$. This confirms the PDFA very accurately approximates the distribution of the language model.

\begin{figure}[htbp]
\centering
\includegraphics[width=0.7\textwidth]{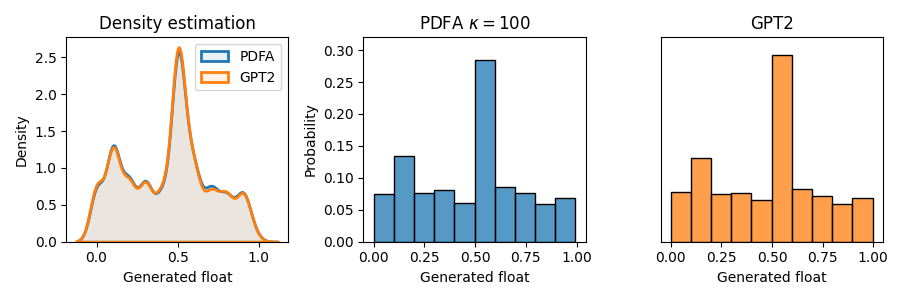}
\includegraphics[width=0.28\textwidth]{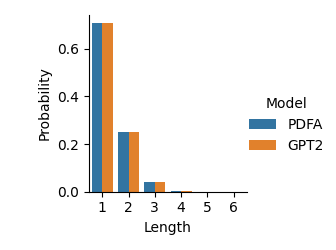}
\caption{Distributions of floats and the lengths of their representing strings (digit sampling).}
\label{fig:hist_kappa_100} 
\end{figure}

\vspace*{-.5em}
\noindent
Fig.~\ref{fig:hist_all_tokens} exhibits the resulting distributions for the second experiment. For 10 bins, the $\chi^2$ pvalue for PDFA vs GPT2 is $0.76$ and for Outlines vs GPT2 is $3\times 10^{-33}$, showing that sampling from the PDFA is more accurate than Outlines for the first digit. However, for 20 bins $\chi^2$ and KS (floats and lengths), pvalues are extremely small. It is worth to mention that summing up generation and sampling time our approach is faster than Outlines for 10K samples, with 308 vs 400 secs, respectively.
\vspace*{-1em}
\begin{figure}[htbp]
\centering
\includegraphics[width=0.7\textwidth]{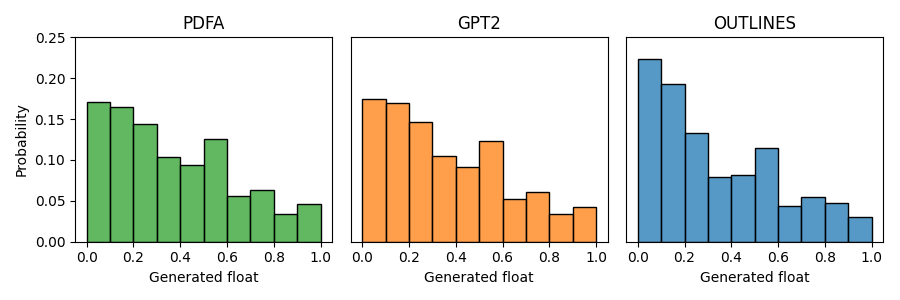}
\includegraphics[width=0.29\textwidth]{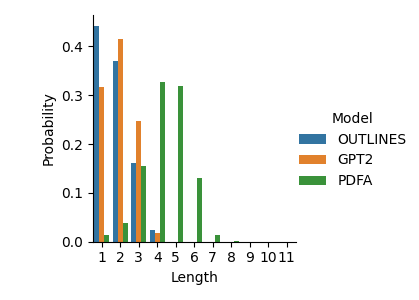}
\caption{Distributions of floats and the lengths of their representing strings (token sampling).}
\label{fig:hist_all_tokens} 
\end{figure}

\section{Conclusions}\label{sec:concl}
\vspace*{-1ex}
This work was motivated by the need of understanding LLM when their operation is controlled by external artifacts, such as grammars, to generate text following a specific format. An important question that arise in this context is how to deal with 0-probabilities that appear when restricting their output.
To start up with, we revised the congruence (\ref{def:indist_DH}) in order to make constructing the quotient less dependent on $\proba{}$ by expressing it in terms of the output of the language model. The first consequence of this operational view is to allow a generalization of the congruence capable of dealing with equivalences on distributions. Besides, it led to developing a variant of the \QuaNT\ active-learning algorithm to efficiently learn PDFA by avoiding to check for 0-probability transitions as much as possible. This is essential to make it computationally feasible by reducing the number of queries to the LLM.
%
%
The experimental results\footnote{ \url{https://github.com/neuralchecker/analyzing_constrained_LLM_through_PDFA_learning}} support the viability of our approach for analyzing and validating statistical properties of LLM, such as bias in text generation. Besides, they provided evidence that distributions resulting from generation of a guided LLM could be well approximated by a learnt PDFA. This opens the door to make these analyses less dependent on sampling by studying properties of the PDFA.

\paragraph{Acknowledgements}\label{sec:aknowledgments}
Research reported in this article has been partially funded by ANII-Agencia Nacional de Investigaci\'on e Innovaci\'on under grants IA\_1\_2022\_1\_173516, FMV\_1\_2023\_1\_175864, POS\_NAC\_2023\_1\_178663, and POS\_FMV\_2023\_1\_1012218.






\bibliography{biblio}
\bibliographystyle{plain}

\appendix
\section{Proof of Proposition~\ref{prop:indist_new}}\label{proof:prop_indist_new}

\begin{proof}
Let $u$ and $v$ in $\Words$ be arbitrary.
\begin{enumerate}
\item Assume that $u\indist v$.

If $\isdef{\lmodel}(\woru)=0$, then the lhs of (\ref{def:indist_DH}) is undefined for any $\worw\in\Words$. Then $\isdef{\lmodel}(\worv)=0$ since otherwise the rhs of (\ref{def:indist_DH}) would be a number for any $\worw\in\Words$ (for instance it equals 1 for $\worw=\emptyW$). By symmetry if $\isdef{\lmodel}(\worv)=0$ then $\isdef{\lmodel}(\woru)=0$. Therefore $\isdef{\lmodel}(\woru)=\isdef{\lmodel}(\worv)$.

Moreover, if $\isdef{\lmodel}(\woru)=\isdef{\lmodel}(\worv)=0$ then $\isdef{\lmodel}(\woru\worw)=\isdef{\lmodel}(\worv\worw)=0$ for all $\worw\in\Words$ and there is nothing more to check.

Suppose that $\isdef{\lmodel}(\woru)=\isdef{\lmodel}(\worv)=1$ so that both sides of (\ref{def:indist_DH}) are defined for any $\worw\in\Words$. Notice also that (\ref{def:indist_DH}) implies $\isdef{\lmodel}(\woru\worw)=\isdef{\lmodel}(\worv\worw)$ for all $\worw\in\Words$. By definition of $\proba{}$ we can rewrite (\ref{def:indist_DH}) as follows:
\begin{equation}\label{eq:product}
\prod_{i=1}^{|\worw|} \lmodel\left(\woru\,\worw_{<i}\right)(\worw_i)
=
\prod_{i=1}^{|\worw|} \lmodel(\worv\,\worw_{<i})(\worw_i)
\end{equation}
for any $\worw\in \Words$ with length $|\worw|\geq 1$. In particular, varying $\worw=\symb\in\Symb$ in (\ref{eq:product}) and noticing that $\lmodel(\woru)$ and $\lmodel(\worv)$ are distributions over $\SymbT$, we see that $\lmodel(\woru)=\lmodel(\worv)$.

We will now prove by induction on the length $|\worw|$ that $\lmodel(\woru\worw)=\lmodel(\worv\worw)$ whenever $\isdef{\lmodel}(\woru\worw)=\isdef{\lmodel}(\worv\worw)=1$. We already proved the claim for $|w|=0$, so suppose it holds true for length $\leq n$. Let $\worw$ be such that $|\worw|=n+1$ and let $\symb\in \Symb$ be such that $\isdef{\lmodel}(\woru\worw\symb)=\isdef{\lmodel}(\worv\worw\symb)=1$. Since all terms involving the products in (\ref{eq:product}) are positive, and by induction hypothesis $\lmodel(\woru\,\worw_{<i})=\lmodel(\worv\,\worw_{<i})$ for all $i=1,\ldots,n$, all the these terms cancel out leaving the equality $\lmodel(\woru\worw)(\symb)=\lmodel(\worv\worw)(\symb)$. Since $\symb\in\Symb$ is arbitrary and $\lmodel(\woru\worw)$ and $\lmodel(\worv\worw)$ are probability distributions, we see again that they must be equal. This completes the proof.

\item Assume $\isdef{\lmodel}(\woru)=\isdef{\lmodel}(\worv)$ and $\forall \worw\in\Words.\ \isdef{\lmodel}(\woru\worw)=\isdef{\lmodel}(\worv\worw)=1 \implies \lmodel(\woru\worw) = \lmodel(\worv\worw)$.

If $\isdef{\lmodel}(\woru)=\isdef{\lmodel}(\worv)=0$, then the quotients in (\ref{def:indist_DH}) are undefined and equality holds trivially for all $\worw\in\Words$.

Let us suppose then that $\isdef{\lmodel}(\woru)=\isdef{\lmodel}(\worv)=1$. We first prove that $\isdef{\lmodel}(\woru\worw)=\isdef{\lmodel}(\worv\worw)$ for all $\worw\in\Words$. In fact, if on the contrary there exists $w\in\Words$ so that $\isdef{\lmodel}(\woru\worw)\neq\isdef{\lmodel}(\worv\worw)$, then there exists $\worw'\in\pref(\worw)$ with $\isdef{\lmodel}(\woru\worw')=\isdef{\lmodel}(\worv\worw')=1$ but $\lmodel(\woru\worw')\neq\lmodel(\worv\worw')$ because they have different support. This contradicts our assumption.

Let $\worw\in\Words$ be so that $\isdef{\lmodel}(\woru\worw)=\isdef{\lmodel}(\worv\worw)=1$.
Then for all prefix $\worw_{<i}$ we also have $\isdef{\lmodel}(\woru\worw_{<i})=\isdef{\lmodel}(\worv\worw_{<i})=1$, and therefore $\lmodel(\woru\worw_{<i})=\lmodel(\worv\worw_{<i})$. In particular, all the terms in (\ref{eq:product}) are equal and therefore (\ref{def:indist_DH}) holds.

This completes the proof that $\woru\indist\worv$.
\end{enumerate}
\end{proof}

\section{Proof of Proposition~\ref{prop:congruence_stasim}}\label{proof:prop_congruence_stasim}

\begin{proof}
Let $\woru\indist_{\Stasim}\worv$. If $\isdef{\lmodel}(\woru)=\isdef{\lmodel}(\worv)=0$, then $\isdef{\lmodel}(\woru\worw)=\isdef{\lmodel}(\worv\worw)=0$ for all $\worw\in\Words$. Then $\woru\symb\indist_{\Stasim}\worv\symb$ trivially.

Suppose now that $\isdef{\lmodel}(\woru)=\isdef{\lmodel}(\worv)=1$ and let $\symb\in\Symb$. We have $\isdef{\lmodel}(\woru\symb)=\isdef{\lmodel}(\worv\symb)$ by Req.~\ref{def:supp}. Let $\worw\in\Words$ be arbitrary, since concatenation of strings is associative, if $\isdef{\lmodel}((\woru\symb)\worw)=\isdef{\lmodel}((\worv\symb)\worw)=1$, then $\isdef{\lmodel}(\woru(\symb\worw))=\isdef{\lmodel}(\worv(\symb\worw))=1$ and by assumption $\lmodel(\woru(\symb\worw))=_{\Stasim} \lmodel(\worv(\symb\worw))$. Thus $\lmodel((\woru\symb)\worw)=_{\Stasim} \lmodel((\worv\symb)\worw)$. This proves that $\woru\symb\indist_{\Stasim}\worv\symb$.
\end{proof}

\section{Proof of Proposition~\ref{prop:one_to_one}}\label{proof:prop_one_to_one}

\begin{proof}
Let $\alpha:\ccla{\Words}_{\Stasim}\setminus\{\boldsymbol{0}\}\to \Words$ be any function satisfying $\alpha(c)\in c$ for all $c\in \ccla{\Words}_{\Stasim}\setminus\{\boldsymbol{0}\}$. In other words, $\{\alpha(c): c \in \ccla{\Words}_{\Stasim}\setminus\{\boldsymbol{0}\}\}$ is a set of representatives of the classes. Let $\beta:\Words\to \cclaICGI{\Words}_\Stasim$ be the quotient map $\beta(\woru)=\cclaICGI{u}_\Stasim$. Define $\phi=\beta\circ \alpha$.

Let $c,c'\in\ccla{\Words}_\Stasim\setminus\boldsymbol{0}$ be such that $\phi(c)=\phi(c')$. Denote $\woru=\alpha(c)$ and $\worv=\alpha(c')$. By construction $\isdef{\lmodel}(\woru)=\isdef{\lmodel}(\worv)=1$ and by Def.~\ref{def:indist_old} we have $\lmodel(\woru\worw)=_\Stasim\lmodel(\worv\worw)$ for all $\worw\in\Words$. In particular $\woru\indist_\Stasim\worv$, or equivalently $c=\ccla{\woru}_\Stasim=\ccla{\worv}_\Stasim=c'$.
\end{proof}

\section{Proof of Proposition~\ref{prop:counterexample}}\label{proof:counterexample}

\begin{proof}
If $\isdef{\lmodel}(\worv\worw)=1$ then $\worw'=\worw$. 
Otherwise, there exists $\worw'\symb\in\pref(\worw)$ such that $1 = \isdef{\lmodel}(\worv\worw') \neq \isdef{\lmodel}(\worv\worw'\symb) = 0$.
Hence, $\supp(\lmodel(\woru\worw')) \neq \supp(\lmodel(\worv\worw'))$
because $\isdef{\lmodel}(\woru\worw'\symb)=1$.
Thus, by Req.~\ref{def:supp}, $\lmodel(\woru\worw')\neq_\Stasim\lmodel(\worv\worw')$.
\end{proof} 

\clearpage
\section{PDFA}
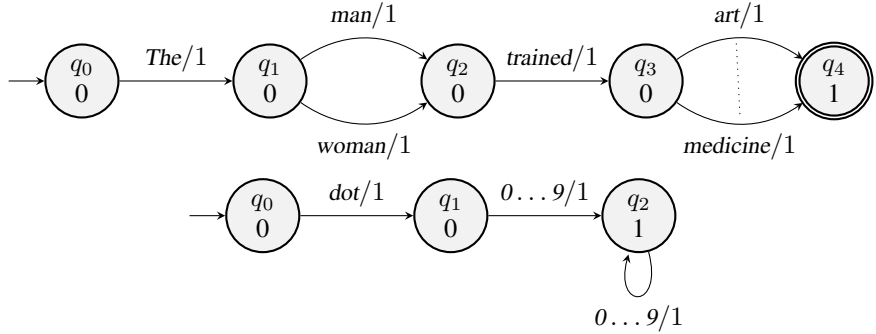
\begin{figure}[htbp]
\centering
\begin{subfigure}[c]{\textwidth}\centering
            \begin{tikzpicture}[node distance=2.5cm]
                \node[state, initial] (q0) {\stackanchor{$q_0$}{0}};
                \node[state] (q1) [right of=q0] {\stackanchor{$q_1$}{0}};
                \node[state] (q2) [right of=q1] {\stackanchor{$q_2$}{0}};
                \node[state] (q3) [right of=q2] {\stackanchor{$q_3$}{0}};
                \node[state, accepting] (q4) [right of=q3] {\stackanchor{$q_4$}{1}};
                \draw[dotted, -] (8.7,0.5) -- (8.75,-0.5);            
                \draw 
                      (q0) edge[above] node {$\charsymb{The}/1$} (q1)
                      (q1) edge[bend left=35, above] node {$\charsymb{man}/1$} (q2)
                      (q1) edge[bend right=35, below] node {$\charsymb{woman}/1$} (q2)
                      (q2) edge[above] node {$\charsymb{trained}/1$} (q3)
                      (q3) edge[bend left=35, above] node {$\charsymb{art}/1$} (q4)
                      (q3) edge[bend right=35, below] node {$\charsymb{medicine}/1$} (q4);
            \end{tikzpicture}
\end{subfigure}
\begin{subfigure}[c]{\textwidth}\centering
            \begin{tikzpicture}[node distance=2.5cm]
                \node[state, initial left] (q0) {\stackanchor{$q_0$}{0}};
                \node[state, right of=q0] (q1) {\stackanchor{$q_1$}{0}};
                \node[state, right of=q1] (q2) {\stackanchor{$q_2$}{1}};
                \draw   
                    (q0) edge[above] node{$\charsymb{dot}/1$} (q1)
                    (q1) edge[above] node{$\charsymb{0}\ldots\charsymb{9}/1$} (q2)
                    (q2) edge[loop below] node{$\charsymb{0}\ldots\charsymb{9}/1$} (q2);
            \end{tikzpicture}
\end{subfigure}
\caption{Guiding automata:%
(above) $\AutG_1$ for man-woman case study
(below) $\AutG_2$ for digits case study} 
\label{fig:guiding_pdfa}
\end{figure}

\begin{figure}[htbp]
\centering
        \begin{tikzpicture}[node distance=3cm]
            \node[state, initial] (q0) {\stackanchor{$q_0$}{0}};
            \node[state] (q1) [right of=q0] {\stackanchor{$q_1$}{0}};
            \node[state] (q21) [above right of=q1] {\stackanchor{$q_{21}$}{0}};
            \node[state] (q22) [below right of=q1] {\stackanchor{$q_{22}$}{0}};
            \node[state] (q31) [right of=q21] {\stackanchor{$q_{31}$}{0}};
            \node[state] (q32) [right of=q22] {\stackanchor{$q_{32}$}{0}};
            \node[state, accepting] (q4) [right of=q3] {\stackanchor{$q_4$}{1}};
            \draw 
                  (q0) edge[above] node {$\charsymb{The}/1$} (q1)
                  (q1) edge[above, pos=.75, left=3ex] node {$\charsymb{man}/0.82$} (q21)
                  (q1) edge[below, pos=.75, left=3ex] node {$\charsymb{woman}/0.18$} (q22)
                  (q21) edge[above] node {$\charsymb{trained}/1$} (q31)
                  (q22) edge[above] node {$\charsymb{trained}/1$} (q32)
                  (q31) edge[bend left=35, above, pos=.4, right=3ex] node {$\charsymb{engineering}/0.43$} (q4)
                  (q31) edge[bend right=35, below, pos=.4, left=3ex] node {$\charsymb{medicine}/0.57$} (q4)
                  (q32) edge[bend left=35, above, pos=.4, left=3ex] node {$\charsymb{business}/0.35$} (q4)
                  (q32) edge[bend right=35, below, pos=.4, right=3ex] node {$\charsymb{medicine}/0.65$} (q4);
        \end{tikzpicture}
\caption{PDFA learnt for man-woman case study (with prefix space tokenizer)}
\label{fig:learnt_man_woman}
\end{figure}
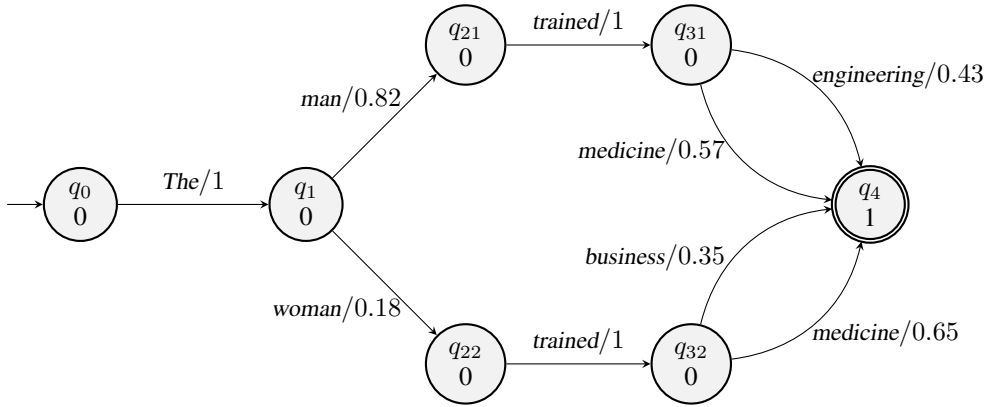

\section{Existence of the probability measure $\probaP$}

\begin{prop}\label{prop:existence_probability}
Let $\lmodel:\Words\to\ProbT$ be a language model. There exists a unique probability measure $\probaP$ in $\Words\cup\Sigma^\omega$ such
\[
\proba{}(\worw)
=
\probaP
\big\{
\worx\in\Words\cup\Sigma^\omega:
\worw\in\pref(\worx)
\big\}
\text{ and }
\proba{\terminal}(\worw)
=
\probaP
\big\{
\worw
\big\}
\]
that for all $\worw\in\Words$.
\end{prop}
\begin{proof}
We first extend the definition of $\lmodel$ in order to include the termination symbol. Let $\lmodelT:\WordsT\to \ProbT$ be defined as follows
\[
\lmodelT(\worw)
=
\begin{cases}
    \lmodel(\worw) & \text{if } \worw\in\Words\\
    \delta_\terminal & \text{if }\worw\in\WordsT\setminus\Words
\end{cases}
\]
where $\delta_\terminal(\symb)=0$ for all $\symb\in\Symb$ and $\delta_\terminal(\terminal)=1$. For each $k\geq 1$, we define the finite dimensional distribution $\probaP_k:\SymbT^k\to [0,1]$ as
\[
\probaP_k
\left[
\worw
\right]
=
\prod_{i=1}^k
\lmodelT(\worw_{<i})(\worw_i)
\]
where we denote $\worw_{<i}=\symb_1\cdots\symb_{i-1}$ if $\worw=\symb_1\cdots\symb_k$, with the convention that $\worw_{<1}=\emptyW$ the empty string. Let us show that $\{\probaP_k\}_{k\geq 1}$ is a consistent family of finite dimensional distributions:
\[
\begin{aligned}
    \sum_{\symb_{k+1}}
    \probaP_{k+1}(\worw \symb_{k+1})
    = \sum_{\symb_{k+1}}\probaP_k[\worw]\lmodelT(w)(\symb_{k+1})
    = \probaP_k[\worw]\sum_{\symb_{k+1}}\lmodelT(w)(\symb_{k+1}) = \probaP_k[\worw] 
\end{aligned}
\]
By the Kolmogorov Extension Theorem (see \cite{shields1996} Thm.I.1.2) there exists a unique probability measure $\probaP$ in $\SymbT^\omega$ such that $\probaP\big\{\worx\in \SymbT^\omega:\worw\in\pref(\worx)\big\}=\probaP_k[\worw]$ for all $k\geq 1$ and any $\worw\in\SymbT^k$. Notice that in the usual measure theoretic terminology the event $\{\worx\in \SymbT^\omega:\worw\in\pref(\worx)\}$ is called a \emph{cylinder}.

The event $A = \left\{\worx\in \SymbT^\omega:\forall k\geq 1\text{ if }\worx_k=\terminal \text{ then }\worx_{k+1}=\terminal\right\}$ can be identified with $\Words\cup \Symb^\omega$. Let us show that $\probaP$ concentrates its measure in $A$, i.e. $\probaP[A]=1$. The complement of $A$ is 
\[
B=
\bigcup_{k=1}^\infty
B_k,
\quad
B_k=
\{
\worx\in\SymbT^\omega:
\worx_k=\terminal\text{ and }\worx_{k+1}\neq\terminal
\}
\]
and $B_k$ is the finite disjoint union of the cylinders of the form $C_{\worw,\symb}=\{\worx\in\SymbT^\omega:\worw\terminal \symb \in\Pref(\worx)\}$ with $\worw\in\SymbT^{k-1}$ and $\symb\in\Symb$. Therefore
\[
\probaP
\left[
B_k
\right]
=
\sum_{\worw,\symb}
\probaP[C_{\worw,\symb}]
=
\sum_{\worw,\symb}
\probaP_{k+1}[C_{\worw,\symb}]
=
\sum_{\worw,\symb}
\probaP_{k-1}\left[\worw\right]
\lmodelT(\worw)(\terminal)
\cancelto{0}{\delta_{\terminal}(\symb)}
=0
\]
and the union bound shows that $\probaP[B]\leq \sum_{k=1}^\infty\probaP[B_k]=0$.

Let us show the link beteween $\probaP$ and $\proba{}$. Let us consider first a string $\worw\in\Words$ of length $k\geq 1$. Since the event $\{\worx\in\Words\cup\Symb^\omega:\worw\in\pref(\worx)\}$ equals the cylinder $C_k=\{\worx\in\SymbT^\omega:\worw\in\pref(\worx)\}$ intersected with $A$, and $A$ has probability one, we have
\[
\begin{aligned}
\probaP
\big\{
\worx\in\Words\cup\Symb^\omega:\worw\in\pref(\worx)
\big\}
&
=
\probaP_k[C_k]
=\prod_{i=1}^k \lmodelT(\worw_{<i})(\worw_i)
=\prod_{i=1}^k \lmodel(\worw_{<i})(\worw_i)
=\proba{}(\worw)
\end{aligned}
\]
In the case $\worw=\emptyW$, the event $\{\worx\in\Words\cup\Symb^\omega:\worw\in\pref(\worx)\}$ equals $A$ and its probability is therefore 1 as it is the case for $P(\emptyW)$.

Finally, let us compute the probability of occurrence of a given finite string $\worw\in\Words$. This string corresponds to the infinite sequence $\worw\terminal\terminal\terminal\cdots$ in $\SymbT^\omega$, which in turn equals the decreasing intersection of the cylinders $C_{\worw,n}=\{\worx\in\SymbT^\omega:w(\terminal)^n\in\pref(\worx)\}$. Therefore
\[
\begin{aligned}
\probaP
\big\{
\worw
\big\}
&
=
\probaP
\left[
\bigcap_{n\geq 1} C_{\worw, n}
\right]
=
\lim_{n\to+\infty}
\left[\prod_{i=1}^{|\worw|} \lmodel(\worw_{<i})(\worw_i)\right]
\lmodel(\worw)(\terminal)
\cancelto{1}{\left[\prod_{j=0}^{n-1}\delta_\terminal(\terminal)\right]}
=
\proba{\terminal}(\worw)
\end{aligned}
\]
This concludes the proof.
\end{proof}

\end{document}